\documentclass[submission,copyright]{eptcs}%,creativecommons
\providecommand{\event}{WORDS 2011}
\usepackage{cite}
\usepackage[english]{babel} 
\usepackage{breakurl}             
\usepackage {array}
\usepackage{gastex}
\usepackage{amssymb,amsfonts,amsmath,amsthm}
\usepackage[usenames]{color}
\newtheorem {teo} {Theorem}
\newtheorem {Rmk} {Remark}
\newtheorem {Lemma} {Lemma}

\DeclareMathOperator{\per}{per}
\DeclareMathOperator{\parent}{parent}
\def\A{{\mathcal A}}
\def\B{{\mathcal B}}
\def\R{{\mathcal R}}
\def\Cyl{{\sf Cyl}}
\newcommand{\II}{
\mbox{
\unitlength=1pt
\begin{picture}(2,12)(-1,5)
\gasset{Nw=1.5,Nh=1.5,Nfill=y,AHnb=0}
\node(o1)(0,16){}
\node(p1)(0,10){}
\node(q1)(0,4){}
\drawedge(o1,p1){}
\drawedge(p1,q1){}
\end{picture} }}
\newcommand{\XX}{
\mbox{
\unitlength=1pt
\begin{picture}(16,12)(4.5,5)
\gasset{Nw=1.5,Nh=1.5,Nfill=y,AHnb=0}
\node(o1)(5,16){}
\node(o4)(20,16){}
\node(p1)(5,10){}
\node(p2)(10,10){}
\node(p3)(15,10){}
\node(p4)(20,10){}
\node(q2)(10,4){}
\node(q3)(15,4){}
\drawedge(o1,p2){}
\drawedge(p1,q2){}
\drawedge(p4,q3){}
\drawedge(o4,p3){}
\end{picture} }}
\newcommand{\V}{
\mbox{
\unitlength=1pt
\begin{picture}(2,10)(0,10)
\gasset{Nw=1.5,Nh=1.5,Nfill=y,AHnb=0}
\node(o1)(0,16){}
\node(p1)(0,10){}
\drawedge(o1,p1){}
\end{picture} }}
\newcommand{\Pp}{
\mbox{
\unitlength=1pt
\begin{picture}(5,10)(0,10)
\gasset{Nw=1.5,Nh=1.5,Nfill=y,AHnb=0}
\node(o1)(0,16){}
\node(p2)(5,10){}
\drawedge(o1,p2){}
\end{picture} }}
\newcommand{\Ll}{
\mbox{
\unitlength=1pt
\begin{picture}(5,10)(0,10)
\gasset{Nw=1.5,Nh=1.5,Nfill=y,AHnb=0}
\node(o2)(5,16){}
\node(p1)(0,10){}
\drawedge(o2,p1){}
\end{picture} }}
\newcommand{\X}{
\mbox{
\unitlength=1pt
\begin{picture}(5,10)(0,10)
\gasset{Nw=1.5,Nh=1.5,Nfill=y,AHnb=0}
\node(o2)(5,16){}
\node(p1)(0,10){}
\node(o1)(0,16){}
\node(p2)(5,10){}
\drawedge(o2,p1){}
\drawedge(o1,p2){}
\end{picture} }}

\newcommand{\VP}{
\mbox{
\unitlength=0.8mm
\begin{picture}(6,6)%(0,10)
\gasset{Nw=1.5,Nh=1.5,Nfill=y,AHnb=0}
\node[linecolor=White, fillcolor=White](o1)(0,0){}
\node[linecolor=White, fillcolor=White](p1)(0,6){}
\drawedge[linewidth=0.5](o1,p1){}
\end{picture} }}
\newcommand{\N}{
\mbox{
\unitlength=0.8mm
\begin{picture}(6,6)%(0,10)
\gasset{Nw=1.5,Nh=1.5,Nfill=y,AHnb=0}
\node[linecolor=White, fillcolor=White](o1)(-2.9,0){}
\node[linecolor=White, fillcolor=White](p1)(2.9,6){}
\drawedge[linewidth=0.5](o1,p1){}
\end{picture} }}
\newcommand{\K}{
\mbox{
\unitlength=0.8mm
\begin{picture}(6,6)%(0,10)
\gasset{Nw=1.5,Nh=1.5,Nfill=y,AHnb=0}
\node[linecolor=White, fillcolor=White](o1)(2.9,0){}
\node[linecolor=White, fillcolor=White](p1)(-2.9,6){}
\drawedge[linewidth=0.5](o1,p1){}
\end{picture} }}

\title{On Pansiot Words Avoiding 3-Repetitions}
\author{Irina A. Gorbunova
\institute{Ural Federal University\\
Ekaterinburg, Russia}
\email{i.a.gorbunova@gmail.com}
\and Arseny M. Shur
\institute{Ural Federal University\\
Ekaterinburg, Russia}
\email{arseny.shur@usu.ru}
}
\def\titlerunning{On Pansiot Words Avoiding 3-Repetitions}
\def\authorrunning{I. A. Gorbunova, A. M. Shur}
\begin{document}
\maketitle
\hyphenation{Using}

\begin{abstract}
The recently confirmed Dejean's conjecture about the threshold between avoidable and un\-avoidable powers of words gave rise to interesting and challenging problems on the structure and growth of threshold words. Over any finite alphabet with $k\ge 5$ letters, Pansiot words avoiding 3-repetitions form a regular language, which is a rather small superset of the set of all threshold words. Using cylindric and 2-dimensional words, we prove that, as $k$ approaches infinity, the growth rates of complexity for these regular languages tend to the growth rate of complexity of some ternary 2-dimensional language. The numerical estimate of this growth rate is ${\approx}1.2421$.
\end{abstract}

Powers, integral and fractional, are the simplest and most natural repetitions in words. Any repetition over an arbitrary fixed alphabet is characterized by the set of all words over this alphabet, avoiding this repetition. The main question concerning such a set is whether it is finite or infinite. For fractional powers, this question is answered by Dejean's conjecture \cite{Dej}, which is now proved in all cases by the efforts of different authors, see \cite{Pan,Mou,MNC,Car,CR1,CR2,Rao}.

Recall that the \textit{exponent} of a word $w$ is the ratio between its length and its minimal period: $\exp(w)=~\hspace{-1.5mm}|w|/\per(w)$. If $\exp(w)=\beta>1$, then $w$ is a \textit{fractional power} ($\beta$-\textit{power}). It is convenient to treat the notion of $\beta$-power as follows: a word $w$ is a $\beta$-power if $\exp(w)\ge\beta$ while $(|w|{-}1)/\per(w)<\beta$, and a $\beta^+\!$-power if $\exp(w)>\beta$ while $(|w|{-}1)/\per(w)\le\beta$. As usual, $\beta^+$ is treated as a ``number'', covering $\beta$ in the usual $\le$ order. A word is called \textit{$\beta$-free} (where $\beta$ can be a number with plus as well) if it contains no $\beta$-powers as factors. A $\beta$-power is \textit{$k$-avoidable} if the number of $k$-ary $\beta$-free words is infinite. Dejean's conjecture states that a $\beta$-power is $k$-avoidable if and only if
$$
\beta\ge(7/4)^+\text{ and } k=3,\ \ \beta\ge(7/5)^+\text{ and } k=4,\ \text{ or }\beta\ge(k/(k{-}1))^+\text{ and } k=2,k\ge5.
$$
The $(k/(k{-}1))^+\!$-free languages over $k$-letter alphabets, where $k\ge5$, are called \textit{threshold languages}; we denote them by $T_k$. We study structure and growth of these languages, aiming at the asymptotic properties as the size of the alphabet increases.

Any threshold language can be approximated from above by a series of regular languages consisting of words that \textit{locally} satisfy the $(k/(k{-}1))^+\!$-freeness property. Namely, these words avoid all $(k/(k{-}1))^+\!$-powers $w$ such that $|w|-\per(w)\le m$, for some constant $m$. From our previous work \cite{ShGo}, it is clear that the case $m=3$ gives a lot of important structural information about the languages $T_k$. Here we study this case in details, using \textit{cylindric representation} that captures the properties common for considered words over all alphabets.

\section {Preliminaries}

We study finite words and two-sided infinite words (\textit{Z-words}) over finite $k$-letter alphabets $\Sigma_k$ and over some special ternary alphabet introduced below. We also consider 2-dimensional words, which are just finite rectangular arrays of alphabetic symbols. Unlike to some commonly used models of 2-dimensional words (cf.~\cite{2D}), we do not use additional symbols to mark the borders of such a word. Factors of 2-dimensional words are also 2-dimensional words.

A (1- or 2-dimensional) language is \textit{factorial}, if it is closed under taking factors of its words. A word $w$ \textit{avoids} a word $u$ if $u$ is not a factor of $w$. The set of all minimal (with respect to the factor order) words avoided by all elements of a factorial language $L$ is called the {\it antidictionary} of $L$. All 1-dimensional languages with finite antidictionaries are regular.

We denote the antidictionary of the threshold language $T_k$ by $A_k$. A word $u\in A_k$ can be factorized as $u=yzy$, where $|yz|=\per(u)$, $|u|/|yz|>k/(k{-}1)$, and all proper factors of $u$ have the exponent at most $k/(k{-}1)$. If $|y|=m$, we call $u$ an \textit{$m$-repetition}.

The finite set $A_k^{(m)}\subset A_k$ consists of all $r$-repetitions with $r\le m$. The notation $T_k^{(m)}$ is used for the (regular) language with the antidictionary $A_k^{(m)}$. Then, $T_k\subseteq T_k^{(m)}$. Since an infinite regular language contains arbitrary powers of some word, one has $T_k\subset T_k^{(m)}$. Clearly, $T_k=\bigcap_{m=1}^{\infty}T_k^{(m)}$.

The \textit{combinatorial complexity} of a language $L$ is a function $C_L(n)$ which returns the number of words in $L$ of length $n$. This function serves as a natural quantitative measure of $L$. ``Big'' [``small''] languages have exponential [resp., subexponential] complexity. Exponential complexity can be described by means of the \textit{growth rate} $\alpha(L)=\limsup_{n\to\infty}(C_L(n))^{1/n}$ (subexponential complexity is indicated by $\alpha(L)=1$). For factorial languages, classical Fekete's lemma implies 
$$
\alpha(L)=\lim_{n\to\infty}(C_L(n))^{1/n}=\inf_{n\to\infty}(C_L(n))^{1/n}.
$$
The growth rate of $T_k^{(m)}$ approximates the growth rate of $T_k$ from above. It is easy to prove that $\lim_{m\to\infty}\alpha(T_k^{(m)})=\alpha(T_k)$.

For regular languages, the growth rate equals the \textit{index} (spectral radius of the adjacency matrix) of recognizing automaton, providing that this automaton is \textit{consistent} (each vertex belongs to some accepting walk), and either deterministic, or non-deterministic but \textit{unambiguous} (there is at most one walk with the given label between two given vertices); see \cite{Sh1}.

\smallskip
In \cite{Pan}, Pansiot showed how to encode all words from the language $T_k^{(2)}$ with ``characteristic'' words over the alphabet $\{0,1\}$. This encoding played a big role in the proof of Dejean's conjecture; so, we refer to the elements of $T_k^{(2)}$ as to \textit{Pansiot words}. These words can be equivalently defined by the following pair of conditions:
\begin{itemize}
\item[(P1)] two closest occurrences of a letter are on the distance $k{-}1$, $k$, or $k{+}1$;
\item[(P2)] two closest occurrences of a letter are followed by different letters.
\end{itemize}
We also consider \textit{Pansiot Z-words}, which are given by (P1), (P2) as well. Finite factors of Pansiot Z-words are exactly Pansiot words.

Now we introduce \textit{cylindric representation} of Pansiot words. Imagine such a word (finite or infinite) as a rope with knots, which are representing letters. This rope is wound around a cylinder such that the knots at distance $k$ are placed one under another (Fig.~\ref{fig0}, a). By (P1), the knots labeled by two closest occurrences of the same letter appear on two consecutive winds of the rope one under another or shifted by one knot (Fig.~\ref{fig0}, b). If we connect these closest occurrences by ``sticks'', we get three types of such sticks: vertical, left-slanted, and right-slanted (Fig.~\ref{fig0},~b). We associate each letter in a Pansiot word with a stick going up from the corresponding knot, getting an encoding of this word by a \textit{cylindric word} over the ternary alphabet $\Delta=\{\V,\Ll,\Pp\}$. Since the sticks allow one to establish equality of letters in a Pansiot word, such a cylindric word [Z-word] uniquely represents the original word [resp., Z-word] up to the permutation of the alphabet. Note that cylindric words avoid squares of letters in view of (P2). Hence, cylindric Z-words are just infinite sequences of blocks $\V\X$ and $\X$.

\begin{figure}[htb]
\unitlength=0.9mm
\centerline{ 
\begin{picture}(64,60)(0,-3) 
\gasset{Nw=1.3,Nh=1.3,Nfill=y,AHnb=0,ExtNL=y,NLdist=1}
\node[Nw=0,Nh=0](f0)(15,50.4){}
\node[Nw=0,Nh=0](f1)(5,45){}
\node[Nw=0,Nh=0](f2)(55,35){}
\node[Nw=0,Nh=0](f3)(5,29){}
\node[Nw=0,Nh=0](f4)(55,19){}
\node[Nw=0,Nh=0](f5)(5,13){}
\node[Nw=0,Nh=0](f6)(55,3){}
\node[Nw=0,Nh=0](f7)(48,5){}
\put(5,53){\line(0,-1){53}}
\put(55,53){\line(0,-1){53}}
\node(o1)(10,41.8){}
\node(o2)(19.6,37){}
\node(o3)(29.9,34){}
\node(o4)(41,33.1){}
\node(o5)(52,34.3){}
\node[Nfill=n](o6)(46.7,37.2){}
\node[Nfill=n](o7)(35,38.3){}
\node[Nfill=n](o8)(24.1,37){}
\node[Nfill=n](o9)(12.5,32.9){}
\node(p1)(10,25.8){}
\node(p2)(19.6,21){}
\node(p3)(29.9,18){}
\node(p4)(41,17.1){}
\node(p5)(52,18.3){}
\node[Nfill=n](p6)(46.7,21.2){}
\node[Nfill=n](p7)(35,22.3){}
\node[Nfill=n](p8)(24.1,21){}
\node[Nfill=n](p9)(12.5,16.9){}
\node(q1)(10,9.8){}
\node(q2)(19.6,5){}
\node(q3)(29.9,2){}
\node(q4)(41,1.1){}
\node(q5)(52,2.3){}
\drawedge[curvedepth=0.3,dash={2.5 2}{0}](f1,f0){}
\drawedge[curvedepth=-6](f1,f2){}
\drawedge[curvedepth=-6,dash={2.5 2}{0}](f2,f3){}
\drawedge[curvedepth=-6](f3,f4){}
\drawedge[curvedepth=-6,dash={2.5 2}{0}](f4,f5){}
\drawedge[curvedepth=-6](f5,f6){}
\drawedge[curvedepth=-0.3,dash={2.5 2}{0}](f6,f7){}
\put(30,52){\makebox(0,0)[cb]{$\ldots$}}
\put(30,6){\makebox(0,0)[cb]{$\ldots$}}
\end{picture} 
\begin{picture}(64,57)(-4,-3)
\gasset{Nw=1.3,Nh=1.3,Nfill=y,AHnb=0,ExtNL=y,NLdist=1}
\node[Nw=0,Nh=0](f0)(15,50.4){}
\node[Nw=0,Nh=0](f1)(5,45){}
\node[Nw=0,Nh=0](f2)(55,35){}
\node[Nw=0,Nh=0](f3)(5,29){}
\node[Nw=0,Nh=0](f4)(55,19){}
\node[Nw=0,Nh=0](f5)(5,13){}
\node[Nw=0,Nh=0](f6)(55,3){}
\node[Nw=0,Nh=0](f7)(48,5){}
\put(5,53){\line(0,-1){53}}
\put(55,53){\line(0,-1){53}}
\node(o1)(10,41.8){$_a$}
\node(o2)(19.6,37){$_b$}
\node(o3)(29.9,34){$_c$}
\node(o4)(41,33.1){$_d$}
\node(o5)(52,34.3){$_e$}
\node(p1)(10,25.8){$_b$}
\node[NLangle=80](p2)(19.6,21){$_a$}
\node[NLangle=270](p3)(29.9,18){$_c$}
\node[NLangle=270](p4)(41,17.1){$_e$}
\node[NLangle=270,NLdist=0.6](p5)(52,18.3){$_d$}
\node[NLangle=270](q1)(10,9.8){$_b$}
\node[NLangle=270](q2)(19.6,5){$_c$}
\node[NLangle=270](q3)(29.9,2){$_a$}
\node[NLangle=270,NLdist=0.4](q4)(41,1.1){$_d$}
\node[NLangle=270](q5)(52,2.3){$_e$}
\drawedge[curvedepth=0.3,dash={2.5 2}{0}](f1,f0){}
\drawedge[curvedepth=-6](f1,f2){}
\drawedge[curvedepth=-6,dash={2.5 2}{0}](f2,f3){}
\drawedge[curvedepth=-6](f3,f4){}
\drawedge[curvedepth=-6,dash={2.5 2}{0}](f4,f5){}
\drawedge[curvedepth=-6](f5,f6){}
\drawedge[curvedepth=-0.3,dash={2.5 2}{0}](f6,f7){}
\drawedge[linewidth=0.25](o1,p2){}
\drawedge[linewidth=0.25](o2,p1){}
\drawedge[linewidth=0.25](o3,p3){}
\drawedge[linewidth=0.25](o1,p2){}
\drawedge[linewidth=0.25](o4,p5){}
\drawedge[linewidth=0.25](o5,p4){}
\drawedge[linewidth=0.25](p1,q1){}
\drawedge[linewidth=0.25](p2,q3){}
\drawedge[linewidth=0.25](p3,q2){}
\drawedge[linewidth=0.25](p4,q5){}
\drawedge[linewidth=0.25](p5,q4){}
\put(30,52){\makebox(0,0)[cb]{$\ldots$}}
\put(32,6){\makebox(0,0)[cb]{$\ldots$}}
\end{picture} 
}
\vspace*{2mm}
\scriptsize{\hspace*{2.75cm}a) Infinite word on a cylinder ($k=9$)\hspace*{3.0cm}b) Sticks (only visible)}
\caption{\small\sl Cylindric representation of Pansiot words.} \label{fig0}
\end{figure}
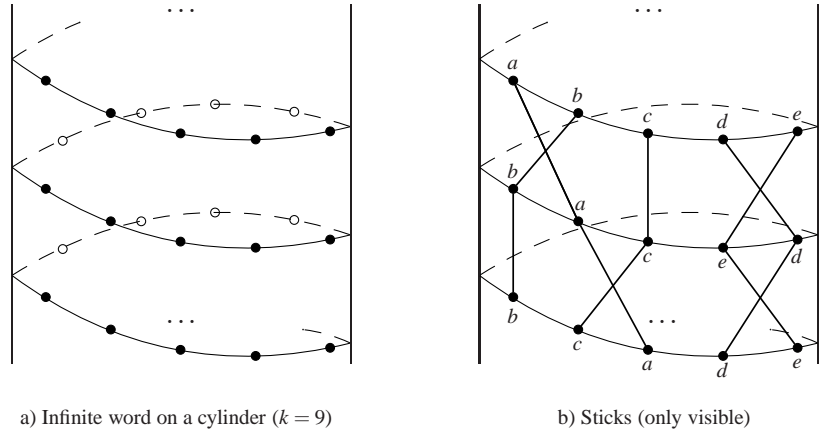
The feature of cylindric words is that they have an additional 2-dimensional structure, allowing one to capture structural properties of Pansiot words through 2-dimensional factors of cylindric words. We say that a Z-word $W$ is \textit{compatible} to a language $L$ if all factors of $W$ belong to $L$.

\begin{teo}[\cite{ShGo}]
For any integer $m\ge3$, there exists a set $S_m$ of 2-dimensional words of size $O(m)\times~\hspace{-1.5mm}O(m)$ over $\Delta$ such that for any $k\ge2m{-}3$, a Pansiot Z-word $W$ over $\Sigma_k$ is compatible to $T_k^{(m)}$ if and only if the corresponding cylindric Z-word has no 2-dimensional factors from $S_m$.
\end{teo}

This theorem states that cylindric words that encode the words from $T_k^{(m)}$ are defined by 2-dimensional avoidance properties. For example, cylindric words of the Pansiot words avoiding 3-repetitions are defined by the avoidance of the structures $\II$ and $\XX$. Indeed, any of these structures implies the existence of three successive letters (say, $a,b$, and $c$) in the encoded Pansiot word such that two occurrences of the factor $abc$ appear one under another at the distance $2k$; since $(2k{+}3)/2k>k/(k{-}1)$, the encoded word contains a 3-repetition. 

For a language $L$, let $\widehat{L}$ be its subset consisting of all factors of Z-words compatible to $L$. By \cite[Theorem~3.1]{Sh2}, $\alpha(\widehat{L})=\alpha(L)$. Let $\Cyl_k^{(m)}$ be the set of all factors of cylindric Z-words encoding Pansiot Z-words compatible to $T_k^{(m)}$. Then clearly $\alpha(\Cyl_k^{(m)})=\alpha(\widehat{T}_k^{(m)})=\alpha(T_k^{(m)})$. Thus, the growth rates of threshold languages can be estimated through the study of cylindric words with simple avoidance properties that are independent of the size of the alphabet. In what follows, we refer to the elements of $\Cyl_k^{(m)}$ as cylindric factors.

\smallskip
The above considerations imply two natural conjectures: \textit{for any fixed $m\ge3$, the sequence $\{\alpha(T_k^{(m)}\!)\}_5^{\infty}$ has a limit as $k$ approaches infinity, and this limit is the ``growth rate'' of the 2-dimensional language defined by the same avoidance properties as $\Cyl_k^{(m)}$}. Through the computations of growth rates for the alphabets with $5,6,\ldots,60$ letters we observed in \cite{ShGo} that the sequence $\{\alpha(T_k^{(3)})\}$ demonstrates fast convergence to the limit ${\approx}1.242096777$.

In this paper, we confirm both conjectures for the case $m=3$. The corresponding 2-dimen\-sional language will be denoted by $D$; it consists of all rectangular words over $\Delta$ having no factors $\genfrac{}{}{0pt}{}{\V}{\V}$ and $\genfrac{}{}{0pt}{}{\Pp\Ll}{\Pp\Ll}$.
In fact, the case $m=3$ is the crucial one to approximate the growth rates of threshold languages, because in \cite{ShGo} it was shown that
\begin{itemize}
\item[-] there is no 4- and 5-repetitions;
\item[-] $m$-repetitions with $m\ge6$ do not affect significantly the growth rate, as far as we can check this by extensive computer-assisted studies based on the results of \cite{Sh1}.
\end{itemize}

\section {Two-dimensional languages}
\textit{Combinatorial complexity} $C_L(n,k)$ of a 2-dimensional language $L$ is the function returning the number of $n\times k$ words in $L$. If $L$ is factorial, then its \textit{growth rate} is defined by the formula 
\begin{equation} \label{ind2d}
\alpha(L)=\lim_{n,k\to\infty}{(C_L(n,k))^{1/nk}}.
\end{equation}
The function $C_L(n,k)$ in this case is submultiplicative for each variable, and hence the existence of the limit (\ref{ind2d}) follows from the multivariate version of Fekete's lemma \cite{Cap}.

On the other hand, it is completely unclear how to calculate the growth rates of 2-dimensional languages. For the 1-dimensional case, the growth rate of a regular language can be found quite efficiently, see \cite{Sh1}. Here we give one idea how to estimate the growth rate of a 2-dimen\-sional language. Since the limit (\ref{ind2d}) exists, we can take any ``diagonal'' subsequence of $C_L(n,k)$; we choose $\big\{\!\!{(C_L(n,n))^{1/n^2}}\big\}_{1}^{\infty}$. Applying Stolz's Theorem (see~\cite{Fih}) twice, we get
\begin{equation*}
\alpha(L)=\lim\limits_{n\to\infty}{(C_L(n,n))^{1/n^2}}=\lim\limits_{n\to\infty}{\left(\frac{C_L(n,n)}{C_L(n{-}1,n{-}1)}\right)^{1/(2n{-}1)}}=\lim\limits_{n\to\infty}\frac{\big({C_L(n,n)C_L(n{-}2,n{-}2)}\big)^{1/2}}{C_L(n{-}1,n{-}1)}
\end{equation*}
if the last two limits exist. Calculating the values of these sequences for the language $D$ (see Table~\ref{tab1}), we see that the last sequence has the best behaviour and allows one to suggest $\alpha(D)\approx 1.2421$. Thus, we get an additional support to the conjecture that $\alpha(D)$ is the limit of the sequence $\big\{\alpha(\Cyl_k^{(3)})\big\}_5^{\infty}$. For the rest of the paper, we set $C(n,k)=C_D(n,k)$.
\begin{table}[htb]
\caption{\small\sl Approximation to the growth rate of the 2-dimensional language $D$.} \label{tab1}
\vspace*{2mm}
\centerline{
\begin{tabular}{|c|c|c|c|}
\hline
$n$&$(C_D(n,n))^{1/n^2}$&$\Big(\frac{C_D(n,n)}{C_D(n{-}1,n{-}1)}\Big)^{1/(2n{-}1)}$&$\frac{\big({C_D(n,n)C_D(n{-}2,n{-}2)}\big)^{1/2}}{C_D(n{-}1,n{-}1)}$\\
\hline
3&1.627251&1.438233&1.191687\\
4&1.525034&1.402991&1.318617\\
5&1.464419&1.362547&1.229958\\
\ldots&\ldots&\ldots&\ldots\\
27&1.280207&1.261332&1.242089\\
28&1.278823&1.260626&1.242080\\
29&1.277537&1.259972&1.242104\\
30&1.276337&1.259362&1.242102\\
\hline
\end{tabular}}
\end{table}

\section {Automata}

Let us fix an arbitrary $k\ge5$. We denote the set of all words of width $k$ from $D$ by $D_k$. It is natural to put $\alpha(D_k)=\lim_{n\to\infty}\!{(C(n,k))^{1/nk}}$; then  $\lim_{k\to\infty}\alpha(D_k)=\alpha(D)$ as the iterative limit of the existing double limit. Note that $D_k$ can be also viewed as a 1-dimensional regular language over the alphabet $\Delta^k$. The automaton $\A$ recognizing $D_k$ can be defined as follows:
\begin{enumerate}
\item[(A1)] the words of length $k$ from $\Cyl_k^{(3)}$ (they coincide with the words of size $1\times k$ from $D_k$) are the vertices; 
\item[(A2)] an edge $u\to v$ exists if and only if the word $\genfrac{}{}{0pt}{}{u}{v}$ of size $2\times k$ belongs to $D_k$; such an edge is labeled by $v$; 
\item[(A3)] each vertex is both initial and terminal.
\end{enumerate}

Note that $\A$ is an unambiguous nondeterministic automaton recognizing $D_k$ as a language over $\Delta^k$. The index of $\A$ (and the growth rate of $D_k$ over $\Delta^k$) equals $\alpha(D_k)^k$. The underlying graph of $\A$ is undirected due to vertical symmetry of the avoided factors. Let $P_u(n)$ be the number of walks of length $n$ in $\A$, starting at the vertex $u$, $P(n)=\sum P_u(n)$ be the number of all walks of length $n$ in $\A$. Then $P(n)=C(n{+}1,k)$.

\smallskip
For the language $\Cyl_k^{(3)}$, we build the \textit{Rauzy graph} $\R$ of order $k{+}1$. The vertices of this graph are the words of $\Cyl_k^{(3)}$ of length $k{+}1$, and a directed edge connects a vertex $u$ to $v$ if and only if some word of $\Cyl_k^{(3)}$ of length $k{+}2$ has the prefix $u$ and the suffix $v$. It is easy to see that the edges of $\R$ can be labeled such that $\R$ becames a deterministic \textit{cover automaton} (all transitions are deterministic, all vertices are both initial and terminal), recognizing the language $\Cyl_k^{(3)}$. Deterministic cover automaton is a special case of unambiguous nondeterministic automaton; so, the index of $\R$ equals $\alpha(\Cyl_k^{(3)})$. Now consider the $k$th power $\R^k$ of $\R$. Note that in most cases the correctness of transition from some vertex $u$ of $\R^k$ to some other vertex $v$ can be checked using only $k$ last symbols of $u$. The only exception is the case when the $k$-letter suffix of $u$ begins and ends with $\X$: if $u$ begins with $\V$, then the $k$-letter suffix of $v$ can begin with both $\V$ and $\Ll$, while if $u$ begins with $\Pp$, then this suffix of $v$ must begin with $\V$ to prevent the appearance of the avoided $2\times2$ factor. Let us require $v$ to begin with $\V$ in any case and consider the automaton $\B$ such that 
\begin{enumerate}
\item[(B1)] the words of length $k$ from $\Cyl_k^{(3)}$ (the suffixes of length $k$ of the vertices from $\R^k$) are the vertices; 
\item[(B2)] an edge $u\to v$ exists if and only if (a) the automaton $\R^k$ contains the edge $au\rightarrow bv$ for some $a,b\in \Delta$, and (b) if $u$ %[resp., $v$] 
has the form $\X\cdots \X$, then $v$ begins %[resp., $u$ ends] 
with $\V$; such an edge is labeled by $v$;
\item[(B3)] each vertex is both initial and terminal.
\end{enumerate}

We will write $P'_u(n)$ for the number of walks of length $n$ in $\B$, starting at $u$, and $P'(n)=\sum P'_u(n)$ for the number of all walks of length $n$ in $\B$. If we denote the number of words of length $nk$ in the language $\Cyl_k^{(3)}$ by $C'(n,k)$, then it is easy to see that $P'(n)\leq C'(n{+}1,k)\leq C(n{+}1,k)$.

\section {Main result}

Since the indices of automata depend only on their adjacency matrices, below we consider the automata $\A$ and $\B$ just as digraphs. Recall that they share the same set of vertices and any edge of $\B$ is contained in $\A$. The outdegrees of a vertex $u$ in $\A$ and ${\B}$ are denoted respectively by $\deg_{\A}^+(u)$ and $\deg_{\B}^+(u)$. We say that the vertices $u$ and $v$ are \textit{similar} if they coincide up to the first 11 letters. Similarity is an equivalence relation; we write $u\sim v$. 

\begin{Rmk}\label{NN}
The classes of $\sim$ are finite, since the cardinality of such a class is the number of words of length $11$ over $\Delta$ that can be extended by the same suffix. The maximum cardinality of such a class is $N=28$ independently of $k$, and is achieved on any suffix that begins with $\Ll$. 
\end{Rmk}

The following two key lemmas hold for any $k\ge12$ (this restriction is necessary only for the existence of 12th symbol in the label of the vertex). 

\begin{Lemma}\label{12th} 
For any vertex $u=u_1\ldots u_k$ and any $a\in\Delta$ such that either $a\ne\V$ or $u_{12}\ne\V$, there exists an edge $u\to x$ in $\B$ such that the 12th letter of $x$ is $a$.
\end{Lemma}

\begin{proof}
Let $x=x_1\cdots x_k$.
We first show that if the condition of the lemma holds for some $i$th letter ($1\leq~\hspace{-1.2mm}i\leq~\hspace{-1.2mm}k$) then it also holds for any $j$th letter ($i<j\leq k$). It suffices to check the case $j=i+1$. Indeed, the minimal structures avoided by the words from $\Cyl_k^{(3)}$ are either factors of length 2, or the ``vertical factor'' $\II$ of height 2, or the ``square factor'' $\XX$ of size $2\times2$. Thus, the possible values of $x_{i+1}$ are determined by $u_{i}$, $u_{i+1}$, and $x_{i}$; each of these values together with $u_{i+1}$ and $u_{i+2}$ determine the possible values of $x_{i+2}$, and so on. There are only four possibilities for the factor $u_iu_{i+1}$. For each of them, we show that if the symbol $x_i$ can take all possible values, then the same is true for $x_{i{+}1}$, see Fig.~\ref{fig1}.

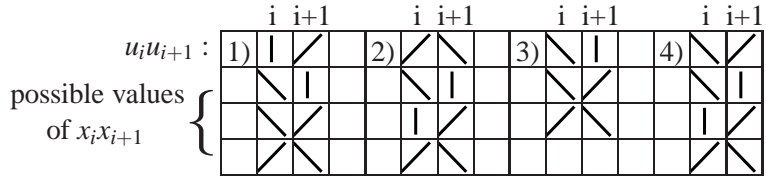
\begin{figure}[htb] 
\unitlength=0.8mm
\centerline{
\begin{picture}(90,28)(-6,0) 
\put(0,0){\N}
\put(6,0){\K}
\put(0,6){\K}
\put(6,6){\N}
\put(0,12){\K}
\put(6,12){\VP}
\put(0,18){\VP}
\put(6,18){\N}
\put(-22,20){$u_iu_{i+1}:$}
\put(-41,12){possible values}
\put(-35,6){of $x_ix_{i+1}$}
\put(-12,6){\Huge\{}
%\put(12,18){$\cdots$}
\put(-5,19){1)}
\put(2,25){i}
\put(6,25){i+1}
\put(24,0){\N}
\put(30,0){\K}
\put(24,6){\VP}
\put(30,6){\N}
\put(24,12){\K}
\put(30,12){\VP}
\put(24,18){\N}
\put(30,18){\K}
%\put(36,18){$\cdots$}
\put(19,19){2)}
\put(26,25){i}
\put(30,25){i+1}
\put(48,6){\N}
\put(54,6){\K}
\put(48,12){\K}
\put(54,12){\N}
\put(48,18){\K}
\put(54,18){\VP}
%\put(60,18){$\cdots$}
\put(43,19){3)}
\put(50,25){i}
\put(54,25){i+1}
\put(72,0){\N}
\put(78,0){\K}
\put(72,6){\VP}
\put(78,6){\N}
\put(72,12){\K}
\put(78,12){\VP}
\put(72,18){\K}
\put(78,18){\N}
%\put(84,18){$\cdots$}
\put(67,19){4)}
\put(74,25){i}
\put(78,25){i+1}
\multiput(-6,0)(6,0){16}%
{\line(0,1){24}}
\multiput(-6,0)(0,6){5}%
{\line(1,0){90}}
\end{picture}}
\vspace*{-1mm}
\caption{\small\sl Proving Lemma~\ref{12th}. If $x_i$ can take any value, $x_{i+1}$ can take any value as well.} \label{fig1}
\end{figure}

In order to prove the lemma we find, for each vertex $u$, the number $i_u$ such that the $i_u$th letter of $x$ can take any value required by the condition of the lemma. If $i_u\le 12$ for any $u$, then we are done with the proof. So we examine all possible beginnings of  
$u$ and try to build the word $x_1\cdots x_{i_u}$ such that $x_{i_u}=a$ for any allowed $a\in\Delta$. Recall that the letter $x_1$ follows $u_k$ in some cylinder word and hence, depends on $u_k$. In order to avoid the consideration of $u_k$ (the restrictions involving $u_k$ depend on $k$), we build the word $x_1\cdots x_{i_u}$ for any $x_1\in\Delta$. The word $x_1\cdots x_{i_u}$ for all $u$ that begin with $\V$ and $\Ll$ is shown in Fig.~\ref{fig2} (cases 1--3 and 4--11, respectively). The maximum value of $i_u$, namely 11, is achieved in case~9. If $u$ begins with $\Pp$, then its factor $u_2\ldots u_{i_u}$ falls into one of the cases 1--11, so, we conclude that $i_u\le12$.
\end{proof}

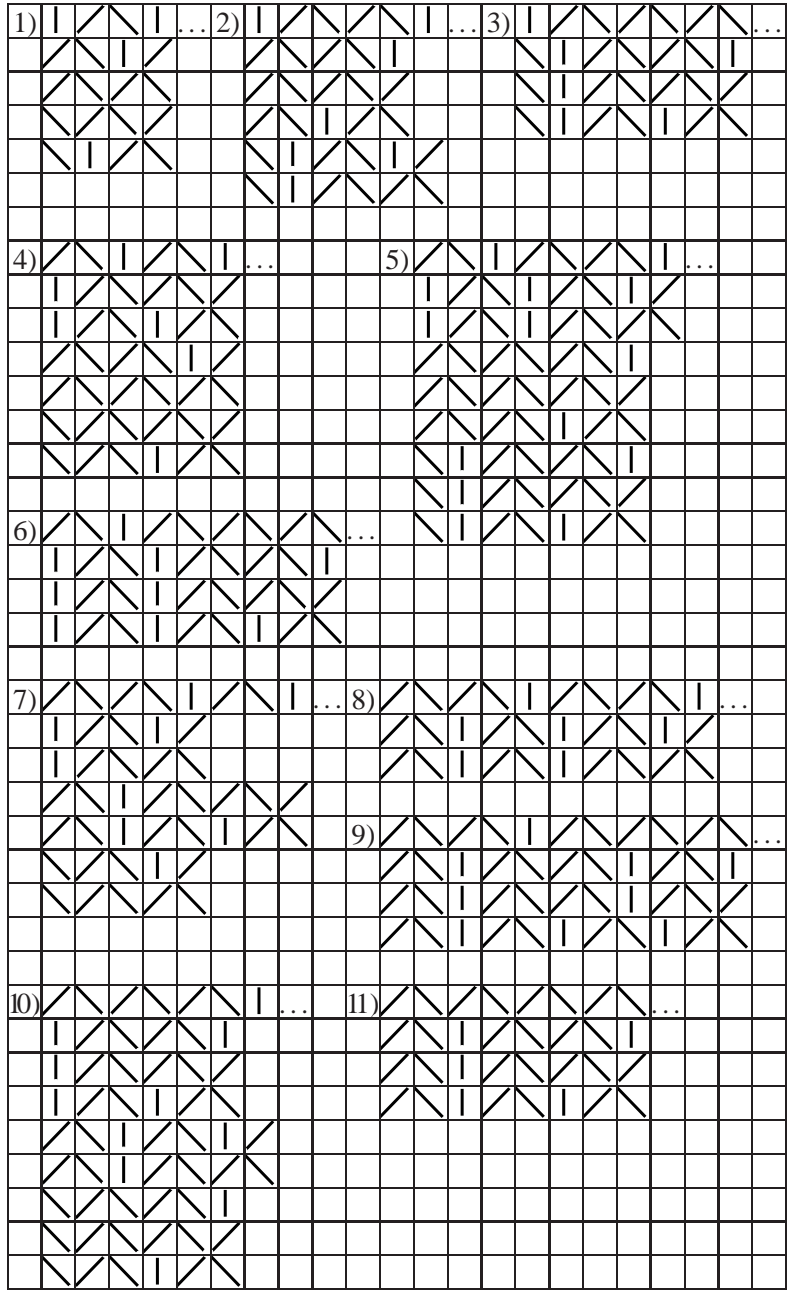
\begin{figure}[phtb] 
\unitlength=0.75mm
\centerline{
\begin{picture}(132,228)(6,-4) 
\put(-5,217){1)}
\put(0,216){\VP}
\put(6,216){\N}
\put(12,216){\K}
\put(18,216){\VP}
\put(24,216){$\cdots$}
\put(0,210){\N}
\put(6,210){\K}
\put(12,210){\VP}
\put(18,210){\N}
%\put(24,210){$\cdots$}
\put(0,204){\N}
\put(6,204){\K}
\put(12,204){\N}
\put(18,204){\K}
%\put(24,204){$\cdots$}
\put(0,198){\K}
\put(6,198){\N}
\put(12,198){\K}
\put(18,198){\N}
%\put(24,198){$\cdots$}
\put(0,192){\K}
\put(6,192){\VP}
\put(12,192){\N}
\put(18,192){\K}
%\put(24,192){$\cdots$}
\put(31,217){2)}
\put(36,216){\VP}
\put(42,216){\N}
\put(48,216){\K}
\put(54,216){\N}
\put(60,216){\K}
\put(66,216){\VP}
\put(72,216){$\cdots$}
\put(36,210){\N}
\put(42,210){\K}
\put(48,210){\N}
\put(54,210){\K}
\put(60,210){\VP}
%\put(66,210){$\cdots$}
\put(36,204){\N}
\put(42,204){\K}
\put(48,204){\N}
\put(54,204){\K}
\put(60,204){\N}
%\put(66,204){$\cdots$}
\put(36,198){\N}
\put(42,198){\K}
\put(48,198){\VP}
\put(54,198){\N}
\put(60,198){\K}
%\put(66,198){$\cdots$}
\put(36,192){\K}
\put(42,192){\VP}
\put(48,192){\N}
\put(54,192){\K}
\put(60,192){\VP}
\put(66,192){\N}
%\put(72,192){$\cdots$}
\put(36,186){\K}
\put(42,186){\VP}
\put(48,186){\N}
\put(54,186){\K}
\put(60,186){\N}
\put(66,186){\K}
%\put(72,186){$\cdots$}
\put(79,217){3)}
\put(84,216){\VP}
\put(90,216){\N}
\put(96,216){\K}
\put(102,216){\N}
\put(108,216){\K}
\put(114,216){\N}
\put(120,216){\K}
\put(126,216){$\cdots$}
\put(84,210){\K}
\put(90,210){\VP}
\put(96,210){\N}
\put(102,210){\K}
\put(108,210){\N}
\put(114,210){\K}
\put(120,210){\VP}
%\put(126,210){$\cdots$}
\put(84,204){\K}
\put(90,204){\VP}
\put(96,204){\N}
\put(102,204){\K}
\put(108,204){\N}
\put(114,204){\K}
\put(120,204){\N}
%\put(126,204){$\cdots$}
\put(84,198){\K}
\put(90,198){\VP}
\put(96,198){\N}
\put(102,198){\K}
\put(108,198){\VP}
\put(114,198){\N}
\put(120,198){\K}
%\put(126,198){$\cdots$}
\put(-5,175){4)}
\put(0,174){\N}
\put(6,174){\K}
\put(12,174){\VP}
\put(18,174){\N}
\put(24,174){\K}
\put(30,174){\VP}
\put(36,174){$\cdots$}
\put(0,168){\VP}
\put(6,168){\N}
\put(12,168){\K}
\put(18,168){\N}
\put(24,168){\K}
\put(30,168){\N}
%\put(36,168){$\cdots$}
\put(0,162){\VP}
\put(6,162){\N}
\put(12,162){\K}
\put(18,162){\VP}
\put(24,162){\N}
\put(30,162){\K}
%\put(36,162){$\cdots$}
\put(0,156){\N}
\put(6,156){\K}
\put(12,156){\N}
\put(18,156){\K}
\put(24,156){\VP}
\put(30,156){\N}
%\put(36,156){$\cdots$}
\put(0,150){\N}
\put(6,150){\K}
\put(12,150){\N}
\put(18,150){\K}
\put(24,150){\N}
\put(30,150){\K}
%\put(36,150){$\cdots$}
\put(0,144){\K}
\put(6,144){\N}
\put(12,144){\K}
\put(18,144){\N}
\put(24,144){\K}
\put(30,144){\N}
%\put(36,144){$\cdots$}
\put(0,138){\K}
\put(6,138){\N}
\put(12,138){\K}
\put(18,138){\VP}
\put(24,138){\N}
\put(30,138){\K}
%\put(36,138){$\cdots$}
\put(-5,127){6)}
\put(0,126){\N}
\put(6,126){\K}
\put(12,126){\VP}
\put(18,126){\N}
\put(24,126){\K}
\put(30,126){\N}
\put(36,126){\K}
\put(42,126){\N}
\put(48,126){\K}
\put(54,126){$\cdots$}
\put(0,120){\VP}
\put(6,120){\N}
\put(12,120){\K}
\put(18,120){\VP}
\put(24,120){\N}
\put(30,120){\K}
\put(36,120){\N}
\put(42,120){\K}
\put(48,120){\VP}
%\put(54,120){$\cdots$}
\put(0,114){\VP}
\put(6,114){\N}
\put(12,114){\K}
\put(18,114){\VP}
\put(24,114){\N}
\put(30,114){\K}
\put(36,114){\N}
\put(42,114){\K}
\put(48,114){\N}
%\put(54,114){$\cdots$}
\put(0,108){\VP}
\put(6,108){\N}
\put(12,108){\K}
\put(18,108){\VP}
\put(24,108){\N}
\put(30,108){\K}
\put(36,108){\VP}
\put(42,108){\N}
\put(48,108){\K}
%\put(54,108){$\cdots$}
\put(61,175){5)}
\put(66,174){\N}
\put(72,174){\K}
\put(78,174){\VP}
\put(84,174){\N}
\put(90,174){\K}
\put(96,174){\N}
\put(102,174){\K}
\put(108,174){\VP}
\put(114,174){$\cdots$}
\put(66,168){\VP}
\put(72,168){\N}
\put(78,168){\K}
\put(84,168){\VP}
\put(90,168){\N}
\put(96,168){\K}
\put(102,168){\VP}
\put(108,168){\N}
%\put(114,168){$\cdots$}
\put(66,162){\VP}
\put(72,162){\N}
\put(78,162){\K}
\put(84,162){\VP}
\put(90,162){\N}
\put(96,162){\K}
\put(102,162){\N}
\put(108,162){\K}
%\put(114,162){$\cdots$}
\put(66,156){\N}
\put(72,156){\K}
\put(78,156){\N}
\put(84,156){\K}
\put(90,156){\N}
\put(96,156){\K}
\put(102,156){\VP}
%\put(108,156){$\cdots$}
\put(66,150){\N}
\put(72,150){\K}
\put(78,150){\N}
\put(84,150){\K}
\put(90,150){\N}
\put(96,150){\K}
\put(102,150){\N}
%\put(108,150){$\cdots$}
\put(66,144){\N}
\put(72,144){\K}
\put(78,144){\N}
\put(84,144){\K}
\put(90,144){\VP}
\put(96,144){\N}
\put(102,144){\K}
%\put(108,144){$\cdots$}
\put(66,138){\K}
\put(72,138){\VP}
\put(78,138){\N}
\put(84,138){\K}
\put(90,138){\N}
\put(96,138){\K}
\put(102,138){\VP}
%\put(108,138){$\cdots$}
\put(66,132){\K}
\put(72,132){\VP}
\put(78,132){\N}
\put(84,132){\K}
\put(90,132){\N}
\put(96,132){\K}
\put(102,132){\N}
%\put(108,132){$\cdots$}
\put(66,126){\K}
\put(72,126){\VP}
\put(78,126){\N}
\put(84,126){\K}
\put(90,126){\VP}
\put(96,126){\N}
\put(102,126){\K}
%\put(108,126){$\cdots$}
\put(-5,97){7)}
\put(0,96){\N}
\put(6,96){\K}
\put(12,96){\N}
\put(18,96){\K}
\put(24,96){\VP}
\put(30,96){\N}
\put(36,96){\K}
\put(42,96){\VP}
\put(48,96){$\cdots$}
\put(0,90){\VP}
\put(6,90){\N}
\put(12,90){\K}
\put(18,90){\VP}
\put(24,90){\N}
%\put(30,90){$\cdots$}
\put(0,84){\VP}
\put(6,84){\N}
\put(12,84){\K}
\put(18,84){\N}
\put(24,84){\K}
%\put(30,84){$\cdots$}
\put(0,78){\N}
\put(6,78){\K}
\put(12,78){\VP}
\put(18,78){\N}
\put(24,78){\K}
\put(30,78){\N}
\put(36,78){\K}
\put(42,78){\N}
%\put(48,78){$\cdots$}
\put(0,72){\N}
\put(6,72){\K}
\put(12,72){\VP}
\put(18,72){\N}
\put(24,72){\K}
\put(30,72){\VP}
\put(36,72){\N}
\put(42,72){\K}
%\put(48,72){$\cdots$}
\put(0,66){\K}
\put(6,66){\N}
\put(12,66){\K}
\put(18,66){\VP}
\put(24,66){\N}
%\put(30,66){$\cdots$}
\put(0,60){\K}
\put(6,60){\N}
\put(12,60){\K}
\put(18,60){\N}
\put(24,60){\K}
%\put(30,60){$\cdots$}
\put(55,97){8)}
\put(60,96){\N}
\put(66,96){\K}
\put(72,96){\N}
\put(78,96){\K}
\put(84,96){\VP}
\put(90,96){\N}
\put(96,96){\K}
\put(102,96){\N}
\put(108,96){\K}
\put(114,96){\VP}
\put(120,96){$\cdots$}
\put(60,90){\N}
\put(66,90){\K}
\put(72,90){\VP}
\put(78,90){\N}
\put(84,90){\K}
\put(90,90){\VP}
\put(96,90){\N}
\put(102,90){\K}
\put(108,90){\VP}
\put(114,90){\N}
%\put(120,90){$\cdots$}
\put(60,84){\N}
\put(66,84){\K}
\put(72,84){\VP}
\put(78,84){\N}
\put(84,84){\K}
\put(90,84){\VP}
\put(96,84){\N}
\put(102,84){\K}
\put(108,84){\N}
\put(114,84){\K}
%\put(120,84){$\cdots$}
\put(55,73){9)}
\put(60,72){\N}
\put(66,72){\K}
\put(72,72){\N}
\put(78,72){\K}
\put(84,72){\VP}
\put(90,72){\N}
\put(96,72){\K}
\put(102,72){\N}
\put(108,72){\K}
\put(114,72){\N}
\put(120,72){\K}
\put(126,72){$\cdots$}
\put(60,66){\N}
\put(66,66){\K}
\put(72,66){\VP}
\put(78,66){\N}
\put(84,66){\K}
\put(90,66){\N}
\put(96,66){\K}
\put(102,66){\VP}
\put(108,66){\N}
\put(114,66){\K}
\put(120,66){\VP}
%\put(126,66){$\cdots$}
\put(60,60){\N}
\put(66,60){\K}
\put(72,60){\VP}
\put(78,60){\N}
\put(84,60){\K}
\put(90,60){\N}
\put(96,60){\K}
\put(102,60){\VP}
\put(108,60){\N}
\put(114,60){\K}
\put(120,60){\N}
%\put(126,60){$\cdots$}
\put(60,54){\N}
\put(66,54){\K}
\put(72,54){\VP}
\put(78,54){\N}
\put(84,54){\K}
\put(90,54){\VP}
\put(96,54){\N}
\put(102,54){\K}
\put(108,54){\VP}
\put(114,54){\N}
\put(120,54){\K}
%\put(126,54){$\cdots$}
\put(-6,43){1\!0)}
\put(0,42){\N}
\put(6,42){\K}
\put(12,42){\N}
\put(18,42){\K}
\put(24,42){\N}
\put(30,42){\K}
\put(36,42){\VP}
\put(42,42){$\cdots$}
\put(0,36){\VP}
\put(6,36){\N}
\put(12,36){\K}
\put(18,36){\N}
\put(24,36){\K}
\put(30,36){\VP}
%\put(36,36){$\cdots$}
\put(0,30){\VP}
\put(6,30){\N}
\put(12,30){\K}
\put(18,30){\N}
\put(24,30){\K}
\put(30,30){\N}
%\put(36,30){$\cdots$}
\put(0,24){\VP}
\put(6,24){\N}
\put(12,24){\K}
\put(18,24){\VP}
\put(24,24){\N}
\put(30,24){\K}
%\put(36,24){$\cdots$}
\put(0,18){\N}
\put(6,18){\K}
\put(12,18){\VP}
\put(18,18){\N}
\put(24,18){\K}
\put(30,18){\VP}
\put(36,18){\N}
%\put(42,18){$\cdots$}
\put(0,12){\N}
\put(6,12){\K}
\put(12,12){\VP}
\put(18,12){\N}
\put(24,12){\K}
\put(30,12){\N}
\put(36,12){\K}
%\put(42,12){$\cdots$}
\put(0,6){\K}
\put(6,6){\N}
\put(12,6){\K}
\put(18,6){\N}
\put(24,6){\K}
\put(30,6){\VP}
%\put(36,6){$\cdots$}
\put(0,0){\K}
\put(6,0){\N}
\put(12,0){\K}
\put(18,0){\N}
\put(24,0){\K}
\put(30,0){\N}
%\put(36,0){$\cdots$}
\put(0,-6){\K}
\put(6,-6){\N}
\put(12,-6){\K}
\put(18,-6){\VP}
\put(24,-6){\N}
\put(30,-6){\K}
%\put(36,-6){$\cdots$}
\put(54,43){1\!1)}
\put(60,42){\N}
\put(66,42){\K}
\put(72,42){\N}
\put(78,42){\K}
\put(84,42){\N}
\put(90,42){\K}
\put(96,42){\N}
\put(102,42){\K}
\put(108,42){$\cdots$}
\put(60,36){\N}
\put(66,36){\K}
\put(72,36){\VP}
\put(78,36){\N}
\put(84,36){\K}
\put(90,36){\N}
\put(96,36){\K}
\put(102,36){\VP}
%\put(108,36){$\cdots$}
\put(60,30){\N}
\put(66,30){\K}
\put(72,30){\VP}
\put(78,30){\N}
\put(84,30){\K}
\put(90,30){\N}
\put(96,30){\K}
\put(102,30){\N}
%\put(108,30){$\cdots$}
\put(60,24){\N}
\put(66,24){\K}
\put(72,24){\VP}
\put(78,24){\N}
\put(84,24){\K}
\put(90,24){\VP}
\put(96,24){\N}
\put(102,24){\K}
%\put(108,24){$\cdots$}
\thinlines
\multiput(-6,-6)(6,0){24}%
{\line(0,1){228}}
\multiput(-6,-6)(0,6){39}%
{\line(1,0){138}}
\end{picture}}
\caption{\small\sl Proving Lemma~\ref{12th}. Cases 1--11 represent different beginnings of the word $u$. Under each beginning, some possible beginnings of the word $x$ are drawn. For each possible first letter of $x$, we exhibit such beginnings ending by all possible letters. In some cases, not all possible beginnings of $x$ are drawn; for such missing beginnings, case 3 refers to case 2, case 6 to case 5, cases 8 and 9 to case 7, and case 11 to case 10.} \label{fig2}
\end{figure}

Lemma~\ref{12th} is used to prove another property of similarity.

\begin{Lemma}\label{sim2} 
If $u\sim v$ and $u\to x$ is an edge in $\A$, then there exists an edge $v\to y$ in $\B$ such that $x\sim y$.
\end{Lemma}

\begin{proof} 
Let $u=u_1\cdots u_k$, $x=x_1\cdots x_k$, $v=v_1\cdots v_k$, and we have to find the vertex $y=y_1\cdots y_k$. Assume that we know only the letters $u_{12},\ldots,u_k$, and $x_{12}$. Then we still can restore all possible values of the factor $x_{13}\cdots x_k$ independently of the letters $u_1,\ldots,u_{11},x_1,\ldots,x_{11}$ (cf. the proof of Lemma~\ref{12th}). 

Now consider all $y$'s such that $v\to y$ is an edge in $\B$ and $y_{12}=x_{12}$. The set of all such $y$'s is nonempty by Lemma~\ref{12th}. Since $v_{12}\cdots v_k=u_{12}\cdots u_k$ by similarity of $u$ and $v$, the set of all possible values of the factor $y_{13}\cdots y_k$ coincides with such a set for the factor $x_{13}\cdots x_k$. Thus, we can pick up $y$ so that the factor $y_{12}\cdots y_k$ equals $x_{12}\cdots x_k$ for the actual value of $x$. Then $x\sim y$, and the lemma is proved.
\end{proof}

\begin{teo} 
The limit $\lim_{k\to\infty}\alpha(T_k^{(3)})$ exists and is equal to $\alpha(D)$.
\end{teo}

\begin{proof} 
Recall that $\alpha(T_k^{(3)})=\alpha(\Cyl_k^{(3)})$. Since the sequence $\big(C_{\Cyl_k^{(3)}}(n)\big)\!^{1/n}$ converges to $\alpha(\Cyl_k^{(3)})$, so does any its subsequence. Hence, $\alpha\big(\Cyl_k^{(3)}\big)=\lim_{n\to\infty}{\big(C'(n,k)\big)\!^{1/nk}}$. On the other hand, we know that $\alpha(D)=\lim_{k\to\infty}\alpha(D_k)=\lim_{k\to\infty}\lim_{n\to\infty}{\big(C(n,k)\big)\!^{1/nk}}$. Thus, let us estimate the ratio $\big({C'(n,k)}/C(n,k))\big)\!^{1/nk}$. The upper bound $\big({C'(n,k)}/C(n,k))\big)\!^{1/nk}\!\leq 1$ is trivial. In order to get the lower bound, we recall that $C(n{+}1,k)=P(n)$ and $C'(n{+}1,k)\geq P'(n)$. 

Let us fix an arbitrary vertex $u$ and consider the $\A$-\textit{tree} (for $u$) defined as follows. The vertices of this tree are labeled by the vertices of $\A$, $u$ being the label of the root. Any vertex labeled by $v$ has $\deg_{\A}^+(v)$ children; the children are labeled by all forward neighbours of $v$ in $\A$. Thus, there is a natural bijection between the set of vertices of level $n$ in the $\A$-tree and the set of all walks from $u$ of length $n$ in the automaton $\A$. That is, $n$th level of the $\A$-tree contains exactly $P_u(n)$ vertices. The $\B$-\textit{tree} is defined in the same way, using $\B$ instead of $\A$. The $n$th level of the $\B$-tree contains $P'_u(n)$ vertices.

Using Lemma~\ref{sim2} inductively, we get that the label of any vertex of $n$th level in the $\A$-tree is similar to the label of some vertex of $n$th level in the $\B$-tree. Let us start from the roots of the trees and inductively construct a total map $\mu$ from the $\A$-tree to the $\B$-tree satisfying the following conditions:
\begin{itemize}
\item[(1)] if $s$ is a level $n$ vertex labeled by $x$, then $\mu(s)$ is a level $n$ vertex labeled by some $y\sim x$;
\item[(2)] $\mu(\parent(s))=\parent(\mu(s))$.
\end{itemize}
The existence of such a map is ensured by Lemma~\ref{sim2} and the structure of trees. 

Now we take a level $n$ vertex $t$ from the $\B$-tree and estimate the size of the set $\mu^{-1}(t)$. Assume that $|\mu^{-1}(\parent(t))|=K$. If $s$ is mapped to $t$, then $\parent(s)\in\mu^{-1}(\parent(t))$. All children of the vertex $\parent(s)$ are different. Hence, by Remark~\ref{NN}, at most $N$ of these children can be mapped to $t$. Thus, $|\mu^{-1}(t)|\le KN$. The case $n=0$ gives us $|\mu^{-1}(t)|=1$ whence we obtain $|\mu^{-1}(t)|\le N^n$. Since $\mu$ is total, we have $P_u(n)\le N^nP'_u(n)$. Summing up these inequalities for all vertices $u$, we finally get $P(n)\le N^nP'(n)$.

Returning to combinatorial complexities, we can write
\begin{gather*}
\frac{1}{N^n}\leq \frac{P'(n)}{P(n)}\leq \frac{C'(n{+}1,k)}{C(n{+}1,k)}\leq 1\,,\\%[4pt]
\left({\frac{1}{N^n}}\right)^{1/(n{+}1)k}\leq \left({\frac{C'(n{+}1,k)}{C(n{+}1,k)}}\right)^{1/(n{+}1)k}\leq 1.
\end{gather*}
Taking the limits of all sides as $n\rightarrow \infty$, we get
$$
\left({\frac{1}{N}}\right)^{1/k}\leq {\frac{\alpha(T_k^{(3)})}{\alpha(D_k)}}\leq 1.
$$
Now we let $k\rightarrow \infty$ and use the squeese theorem to conclude that the limit $\lim_{k\to\infty}{\alpha(T_k^{(3)})/\alpha(D_k)}$ exists and is equal to 1 (recall that $N$ is independent of $k$). Since the limit $\lim_{k\to\infty}\alpha(D_k)=\alpha (D)$ also exists, we have
$$
\alpha(D)=\alpha(D)\cdot1=\lim_{k\to\infty}\alpha(D_k){\cdot}\lim_{k\to\infty}{\frac{\alpha\big(T_k^{(3)}\big)}{\alpha(D_k)}}= \lim_{k\to\infty}{\frac{\alpha\big(T_k^{(3)}\big)}{\alpha(D_k)}}\cdot\alpha(D_k)=\lim_{k\to\infty}{\alpha\big(T_k^{(3)}\big)},
$$
as desired.
\end{proof}

\begin{Rmk}
From the proof of the above theorem it is clear that the actual value of the constant $N$ such that $P(n)\approx N^nP'(n)$ is much smaller than 28. Computations show that $N\approx 2.119$. Hence, the set $D_k$ of 2-dimensional words of width $k$ is not much bigger than the corresponding set $\Cyl_k^{(3)}$ of cylindric words.
\end{Rmk}

%\nocite{*}
\bibliographystyle{eptcs}

\end{document}